\documentclass[a4paper]{extarticle}

\usepackage{hyperref}
\usepackage{fullpage}

\title{Space-Efficient Algorithms for {Longest Increasing Subsequence}%
\footnote{Partially supported by MEXT KAKENHI grant number 24106004. 
}}

\author{Masashi~Kiyomi\thanks{Yokohama City University. Yokohama, Japan. \texttt{masashi@yokohama-cu.ac.jp}} \and
Hirotaka~Ono\thanks{Nagoya University. Nagoya, Japan. \texttt{ono@.nagoya-u.jp}} \and
Yota~Otachi\thanks{Kumamoto University. Kumamoto, Japan. \texttt{otachi@cs.kumamoto-u.ac.jp}} \and
Pascal~Schweitzer\thanks{TU Kaiserslautern. Kaiserslautern, Germany. \texttt{schweitzer@cs.uni-kl.de}} \and
Jun~Tarui\thanks{The University of Electro-Communications. Chofu, Japan.~\texttt{tarui@ice.uec.ac.jp}}}

\newcommand{\seq}[1]{\langle #1 \rangle}
\newcommand{\subseq}[1]{[#1]}

\newcommand{\sseq}{\tau} 

\newcommand{\PS}{\textsc{Patience Sorting}} 
\newcommand{\RPS}{\textsc{Reverse Patience Sorting}} 

\usepackage{amsthm}
\theoremstyle{plain}
\newtheorem{theorem}{Theorem}[section]
\newtheorem{lemma}[theorem]{Lemma}
\newtheorem{proposition}[theorem]{Proposition}
\newtheorem{observation}[theorem]{Observation}
\newtheorem{corollary}[theorem]{Corollary}

\usepackage[table,xcdraw]{xcolor}

\usepackage{color}
\definecolor{lightblue}{rgb}{0.5,0.5,1.0}
\definecolor{darkred}{rgb}{0.8,0,0}
\definecolor{darkgreen}{rgb}{0,0.5,0}
\definecolor{darkblue}{rgb}{0,0,0.5}

\definecolor{vr}{rgb}{0.78,0.08,0.52} 

\usepackage{algorithm}
\usepackage[noend]{algpseudocode}

\renewcommand{\mid}{:}

\newcommand{\lis}{\mathsf{lis}}

\begin{document}

\maketitle

\begin{abstract}
Given a sequence of integers, we want to find a longest increasing subsequence of the sequence.
It is known that this problem can be solved in $O(n \log n)$ time and space.
Our goal in this paper is to reduce the space consumption while keeping the time complexity small.
For $\sqrt{n} \le s \le n$, we present algorithms that use $O(s \log n)$ bits and 
$O(\frac{1}{s} \cdot n^{2} \cdot \log n)$ time for computing the length of a longest increasing subsequence,
and $O(\frac{1}{s} \cdot n^{2} \cdot \log^{2} n)$ time for finding an actual subsequence.
We also show that the time complexity of our algorithms is optimal up to polylogarithmic factors
in the framework of sequential access algorithms with the prescribed amount of space.
\end{abstract}


\section{Introduction}

Given a sequence of integers (possibly with repetitions),
the problem of finding a longest increasing subsequence (LIS, for short) is 
a classic problem in computer science which has many application areas including bioinfomatics and physics
(see \cite{SunW2007communication} and the references therein).
It is known that LIS admits an $O(n \log n)$-time algorithm that uses
$O(n \log n)$ bits of working space~\cite{Schensted1961longest,Fredman1975,AldousP1999longest},
where $n$ is the length of the sequence.

A wide-spread algorithm achieving these bounds is {\PS}, devised by Mallows~\cite{Mallows1962ps,Mallows1963,Mallows1973ps}.
Given a sequence of length $n$, {\PS} partitions
the elements of the sequence into so-called \emph{piles}.
It can be shown that the number of piles coincides with the length of a longest increasing subsequence
(see Section~\ref{sec:ps} for details).
Combinatorial and statistical properties of the piles in {\PS} are well studied 
(see \cite{AldousP1999longest,BursteinL2006combinatorics,Romik2015surprising}).

However, with the dramatic increase of the typical data sizes in applications over the last decade, a main memory consumption in the order of $\Theta(n \log n)$ bits is excessive in many algorithmic contexts, especially for basic subroutines such as LIS.
We therefore investigate the existence of space-efficient algorithms for LIS.\@


\subparagraph{Our results}

In this paper, we present the first space-efficient algorithms for LIS that are exact.
We start by observing that
when the input is restricted to permutations, an algorithm using $O(n)$ bits can be obtained
straightforwardly by modifying a previously known algorithm (see Section~\ref{ssec:n-bits-algorithm}).
Next, we observe that a Savitch type algorithm~\cite{Savitch1970relationships} for this problem uses 
$O(\log^{2} n)$ bits and thus runs in quasipolynomial time.
However, we are mainly interested in space-efficient algorithms that also behave well with regard to running time.
To this end we develop an algorithm that determines the length of a longest increasing subsequence using $O(\sqrt{n} \log n)$ bits which runs in $O(n^{1.5} \log n)$ time. Since the constants hidden in the~O-notation are negligible, 
the algorithm, when executed in the main memory of a standard computer, may handle a peta-byte input on external storage.

More versatile, in fact, our space-efficient algorithm is \emph{memory-adjustable} 
in the following sense. (See~\cite{AsanoEK13priority} for information on memory-adjustable algorithms.)
When a memory bound~$s$ with $\sqrt{n} \le s \le n$ is given to the algorithm, it computes with $O(s \log n)$ bits of working space 
in $O(\frac{1}{s} \cdot n^{2} \log n)$ time the length of a longest increasing subsequence.
When $s = n$ our algorithm is equivalent to the previously known algorithms mentioned above.
When $s = \sqrt{n}$ it uses, as claimed above, $O(\sqrt{n} \log n)$ bits and runs in $O(n^{1.5} \log n)$ time.

The algorithm only determines the length of a longest increasing subsequence. 
To actually find such a longest increasing subsequence, one can run the length-determining algorithm $n$ times to successively construct the sought-after subsequence. 
This would give us a running time of~$O(\frac{1}{s}\cdot n^{3} \log n)$. However, we show that one can do much better, achieving a running time of $O(\frac{1}{s} \cdot n^{2} \log^{2} n)$ 
without any increase in space complexity,
by recursively finding a \emph{near-mid} element of a longest increasing subsequence.

To design the algorithms, we study the structure of the piles arising in {\PS} in depth and show that
maintaining certain information regarding the piles suffices to simulate the algorithm. 
Roughly speaking, our algorithm divides the execution of {\PS} into $O(n/s)$ phases,
and in each phase it computes in $O(n \log n)$ time information on the next $O(s)$ piles, while forgetting previous information.

Finally, we complement our algorithm with a lower bound in a restricted computational model. 
In the \emph{sequential access model}, an algorithm can access the input only sequentially.
We also consider further restricted algorithms in the \emph{multi-pass model},
where an algorithm has to read the input sequentially from left to right
and can repeat this multiple (not necessarily a constant number of) times.
Our algorithm for the length works within the multi-pass model,
while the one for finding a subsequence is a sequential access algorithm.
Such algorithms are useful when large data is placed in an external storage
that supports efficient sequential access.
We show that the time complexity of our algorithms
is optimal up to polylogarithmic factors in these models.


\subparagraph{Related work}
The problem of finding a longest increasing subsequence (LIS) is 
among the most basic algorithmic problems on integer arrays and has been studied continuously since the early
1960's.
It is known that LIS can be solved in $O(n \log n)$ time and space~\cite{Schensted1961longest,Fredman1975,AldousP1999longest},
and that any comparison-based algorithm needs $\Omega(n \log n)$ comparisons
even for computing the length of a longest increasing subsequence~\cite{Fredman1975,Ramanan1997tight}.
For the special case of LIS where the input is restricted to permutations,
there are $O(n \log\log n)$-time algorithms~\cite{HuntS1977fast,BespamyatnikhS2000enumerating,CrochemoreP2010fast}.
{\PS}, an efficient algorithm for LIS, has been a research topic
in itself, especially in the context of Young tableaux~\cite{Mallows1962ps,Mallows1963,Mallows1973ps,AldousP1999longest,BursteinL2006combinatorics,Romik2015surprising}.

Recently, LIS has been studied intensively in the data-streaming model,
where the input can be read only once (or a constant number of times) sequentially from left to right.
This line of research was initiated by Liben-Nowell, Vee, and Zhu~\cite{Liben-NowellVZ2006lis},
who presented an exact one-pass algorithm and a lower bound for such algorithms.
Their results were then improved and extended by many other groups~\cite{GopalanJKK2007estimating,
SunW2007communication,GalG2010lower,SaksS2013space,
ErgunJ2015monotonicity,NaumovitzS2015polylogarithmic,SaksS2017estimating}.
These results give a deep understanding on streaming algorithms with a constant number of passes
even under the settings with randomization and approximation.
(For details on these models,
see the very recent paper by Saks and Seshadhri~\cite{SaksS2017estimating} and the references therein.)
On the other hand, multi-pass algorithms with a non-constant number of passes have not been studied for LIS.

While space-limited algorithms on both RAM and multi-pass models
for basic problems have been studied since the early stage of algorithm theory,
research in this field has recently intensified.
Besides LIS, other frequently studied problems include sorting and 
selection~\cite{MunroP1980selection,BorodinC1982time,Frederickson1987upper,PagterR1998optimal},
graph searching~\cite{AsanoIKKOOSTU2014depth,ElmasryHK2015space,PilipczukW2016space,ChakrabortyS2017space}, 
geometric computation~\cite{ChanC2007multi,DarwishE2014optimal,BanyassadyKMvRRSS2017improved,AhnBOS2017time}, and
{$k$-SUM}~\cite{Wang14random,LincolnWWW16deterministic}.


\section{Preliminaries}

Let $\sseq = \seq{\sseq(1),\sseq(2),\ldots,\sseq(n)}$ be a sequence of $n$ integers possibly with repetitions.
For $1 \le i_{1} < \dots < i_{\ell} \le n$,
the \emph{subsequence} $\sseq\subseq{i_{1}, \dots, i_{\ell}}$ of $\sseq$
is the sequence $\seq{\sseq(i_{1}), \dots, \sseq(i_{\ell})}$.
A subsequence $\sseq\subseq{i_{1}, \dots, i_{\ell}}$ is an \emph{increasing subsequence} of $\sseq$
if $\sseq(i_{1}) < \dots < \sseq(i_{\ell})$. 
If $\sseq(i_{1}) \le \dots \le \sseq(i_{\ell})$, then the sequence $\sseq$ is \emph{non-decreasing}.
We analogously define \emph{decreasing subsequences} and \emph{non-increasing subsequences}.
By $\lis(\sseq)$, we denote the length of a longest increasing subsequence of $\sseq$.

For example, consider a sequence $\sseq_{1} = \seq{2, 8, 4, 9, 5, 1, 7, 6, 3}$.
It has an increasing subsequence $\sseq_{1}\subseq{1,3,5,8} = \seq{2, 4, 5, 6}$.
Since there is no increasing subsequence of $\sseq_{1}$ with length~5 or more, we have $\lis(\sseq_{1}) = 4$.

In the computational model in this paper, we use the RAM model with the following restrictions that are standard in the context of sublinear space algorithms.
The input is in a read-only memory and the output must be produced on a write-only memory.
We can use an additional memory that is readable and writable.
Our goal is to minimize the size of the additional memory
while keeping the running time fast.
We measure space consumption in the number of bits used (instead of words)
within the additional memory.


\section{\PS}
\label{sec:ps}

Since our algorithms are based on the classic \PS,
we start by describing it in detail and recalling some important properties regarding its internal configurations.

Internally, the algorithm maintains a collection of piles.
A \emph{pile} is a stack of integers.
It is equipped with the procedures push and top:
the push procedure appends a new element to become the new top of the pile; and
the top procedure simply returns the element on top of the pile, which is always the one that was added last.

We describe how {\PS} computes $\lis(\sseq)$.
See Algorithm~\ref{alg:patience-sorting}.
The algorithm scans the input $\sseq$ from left to right (Line~\ref{alg:ps:for-loop}).
It tries to push each newly read element $\sseq(i)$ to a pile with a top element larger than or equal to $\sseq(i)$.
If on the one hand there is no such a pile, {\PS} creates a new pile to which it pushes $\sseq(i)$ (Line~\ref{alg:ps:new-pile}).
On the other hand, if at least one such pile exists, {\PS} pushes $\sseq(i)$ to the oldest pile that satisfies the property (Line~\ref{alg:ps:oldest-pile}).
After the scan, the number of piles is the output, which happens to be equal to $\lis(\sseq)$ (Line~\ref{alg:ps:return}).

\begin{algorithm}
  \caption{{\PS}}
  \label{alg:patience-sorting}
  \begin{algorithmic}[1]
    \State set $\ell := 0$ and initialize the dummy pile $P_{0}$ with the single element $-\infty$
    \For{$i=1$ \textbf{to} $n$} \label{alg:ps:for-loop}
        \If{$\sseq(i) > \mathtt{top}(P_{\ell})$}
            \State increment $\ell$, let $P_{\ell}$ be a new empty pile, and set $j := \ell$ \label{alg:ps:new-pile}
	\Else
	    \State set $j$ to be the smallest index with $\sseq(i) \le \mathtt{top}(P_{j})$ \label{alg:ps:oldest-pile}
	\EndIf
	\State push $\sseq(i)$ to $P_{j}$
    \EndFor
    \State \Return $\ell$ \label{alg:ps:return}
  \end{algorithmic}
\end{algorithm}

We return to the sequence $\sseq_{1} = \seq{2, 8, 4, 9, 5, 1, 7, 6, 3}$ for an example.
The following illustration shows the execution of Algorithm~\ref{alg:patience-sorting} on $\sseq_{1}$. In each step the bold number is the newly added element.
The colored (and underlined) elements in the final piles form a longest increasing subsequence $\sseq_{1}\subseq{1,3,5,8} = \seq{2, 4, 5, 6}$,
which can be extracted as described below.

\smallskip

{
  \small
  \tabcolsep=0.5mm
  \noindent
  \begin{tabular}{c}
    \\
    \\
    \textbf{2} \\
    \hline
    $P_{1}$
  \end{tabular}
  \hfill
  \begin{tabular}{cc}
    \\
    &\\
    2 & \textbf{8} \\
    \hline
    $P_{1}$ & $P_{2}$
  \end{tabular}
  \hfill
  \begin{tabular}{cc}
    \\
    & \textbf{4} \\
    2 & 8 \\
    \hline
    $P_{1}$ & $P_{2}$
  \end{tabular}
  \hfill
  \begin{tabular}{ccc}
    \\
      & 4 & \\
    2 & 8 & \textbf{9} \\
    \hline
    $P_{1}$ & $P_{2}$ & $P_{3}$
  \end{tabular}
  \hfill
  \begin{tabular}{ccc}
    \\
      & 4 & \textbf{5} \\
    2 & 8 & 9 \\
    \hline
    $P_{1}$ & $P_{2}$ & $P_{3}$
  \end{tabular}
  \hfill
  \begin{tabular}{ccc}
    \\
    \textbf{1} & 4 & 5 \\
    2 & 8 & 9 \\
    \hline
    $P_{1}$ & $P_{2}$ & $P_{3}$
  \end{tabular}
  \hfill
  \begin{tabular}{cccc}
    \\
    1 & 4 & 5 & \\
    2 & 8 & 9 & \textbf{7} \\
    \hline
    $P_{1}$ & $P_{2}$ & $P_{3}$ & $P_{4}$
  \end{tabular}
  \hfill
  \begin{tabular}{cccc}
    \\
    1 & 4 & 5 & \textbf{6} \\
    2 & 8 & 9 & 7 \\
    \hline
    $P_{1}$ & $P_{2}$ & $P_{3}$ & $P_{4}$
  \end{tabular}
  \hfill
  \begin{tabular}{cccc}
      & \textbf{3} &   &   \\
    1 & \textcolor{vr}{\underline{4}} & \textcolor{vr}{\underline{5}} & \textcolor{vr}{\underline{6}} \\
    \textcolor{vr}{\underline{2}} & 8 & 9 & 7 \\
    \hline
    $P_{1}$ & $P_{2}$ & $P_{3}$ & $P_{4}$
  \end{tabular}
}

\begin{proposition}
[\cite{Schensted1961longest,Fredman1975,AldousP1999longest}]
Given a sequence $\sseq$ of length $n$,
{\PS} computes $\lis(\sseq)$ in $O(n \log n)$ time using $O(n \log n)$ bits of working space.
\end{proposition}

\subsection{Correctness of \PS}
It is observed in \cite{BursteinL2006combinatorics} that
when the input is a permutation $\pi$,
the elements of each pile form a decreasing subsequence of $\pi$.
This observation easily generalizes as follows.
\begin{observation}
Given a sequence $\sseq$,
the elements of each pile constructed by {\PS}
form a non-increasing subsequence of $\sseq$.
\end{observation}
Hence, any increasing subsequence of $\sseq$ can contain at most one element in each pile.
This implies that $\lis(\sseq) \le \ell$.

Now we show that $\lis(\sseq) \ge \ell$.
Using the piles, we can obtain an increasing subsequence of length $\ell$,
in reversed order, as follows~\cite{AldousP1999longest}:
\begin{enumerate}
  \item Pick an arbitrary element of $P_{\ell}$;
  \item For $1 \le i < \ell$, let $\sseq(h)$ be the element picked from $P_{i+1}$.
  Pick the element $\sseq(h')$ that was the top element of $P_{i}$ when $\sseq(h)$ was pushed to $P_{i+1}$.
\end{enumerate}
Since $h' < h$ and $\sseq(h') < \sseq(h)$ in each iteration,
the $\ell$ elements that are selected form an increasing subsequence of $\sseq$.
This completes the correctness proof for \PS.

The proof above can be generalized to show the following characterization for the piles. 
\begin{proposition}
[\cite{BursteinL2006combinatorics}]
\label{prop:pile-characterization}
$\sseq(i) \in P_{j}$ if and only if 
a longest increasing subsequence of $\sseq$ ending at $\sseq(i)$ has length $j$.
\end{proposition}

\subsection{Time and space complexity of \PS}

Observe that at any point in time, the top elements of the piles are ordered increasingly from left to right.
Namely, $\mathtt{top}(P_{k}) < \mathtt{top}(P_{k'})$ if $k < k'$.
This is observed in \cite{BursteinL2006combinatorics} for inputs with no repeated elements.
We can see that the statement holds also for inputs with repetitions.
\begin{observation}
  \label{obs:increasing-top}
  At any point in time during the execution of {\PS} and for any $k$ and $k'$ with $1 \le k < k' \le \ell$,
  we have $\mathtt{top}(P_{k}) < \mathtt{top}(P_{k'})$ if $P_{k}$ and $P_{k'}$ are nonempty.
\end{observation}
\begin{proof}
We prove the statement by contradiction. Let $i$ be the first index for which {\PS} pushes $\sseq(i)$ to some pile $P_{j}$,
so that the  statement of the observation becomes false.

First assume that $\mathtt{top}(P_{j}) \ge \mathtt{top}(P_{j'})$ for some $j' > j$.
Let $\sseq(i')$ be the element in $P_{j}$ pushed to the pile right before $\sseq(i)$.
By the definition of ${\PS}$, it holds that 
\[
  \sseq(i') \ge \sseq(i) = \mathtt{top}(P_{j}) \ge \mathtt{top}(P_{j'}).
\]
This contradicts the minimality of $i$
because $\sseq(i')$ was the top element of $P_{j}$ before $\sseq(i)$ was pushed to $P_{j}$.

Next assume that $\mathtt{top}(P_{j'}) \ge \mathtt{top}(P_{j})$ for some $j' < j$.
This case contradicts the definition of {\PS}
since $\sseq(i) = \mathtt{top}(P_{j}) \le \mathtt{top}(P_{j'})$ and thus $\sseq(i)$ actually has to be pushed to a pile
with an index smaller or equal to $j'$.
\end{proof}
The observation above implies that
Line~\ref{alg:ps:oldest-pile} of Algorithm~\ref{alg:patience-sorting} can be executed in $O(\log n)$ time by using binary search.
Hence, {\PS} runs in $O(n \log n)$ time.

The total number of elements in the piles is $O(n)$ and thus {\PS} consumes $O(n \log n)$ bits.
If it maintains all elements in the piles,
it can compute an actual longest increasing subsequence in the same time and space complexity as described above.
Note that to compute $\lis(\sseq)$, it suffices to remember the top elements of the piles.
However, the algorithm still uses $\Omega(n \log n)$ bits when $\lis(\sseq) \in \Omega(n)$.

\subsection{A simple $O(n)$-bits algorithm}
\label{ssec:n-bits-algorithm}
Here we observe that, when the input is a permutation $\pi$ of $\{1,\dots,n\}$,
$\lis(\pi)$ can be computed in $O(n^{2})$ time with $O(n)$ bits of working space.
The algorithm maintains a used/unused flag for each number in $\{1,\dots,n\}$.
Hence, this noncomparison-based algorithm cannot be generalized for general inputs directly.

Let $\sseq$ be a sequence of integers without repetitions.
A subsequence $\sseq\subseq{i_{1}, \dots, i_{\ell}}$ is the \emph{left-to-right minima subsequence}
if $\{i_{1}, \dots, i_{\ell}\} = \{i \mid \sseq(i) = \min\{\sseq(j) \mid 1 \le j \le i \}\}$.
In other words, the left-to-right minima subsequence is made by scanning $\sseq$ from left to right
and greedily picking elements to construct a maximal decreasing subsequence.

Burstein and Lankham~\cite[Lemma~2.9]{BursteinL2006combinatorics}
showed that the first pile $P_{1}$ is the left-to-right minima subsequence of $\pi$
and that the $i$th pile $P_{i}$ is the left-to-right minima subsequence of a sequence obtained from $\pi$
by removing all elements in the previous piles $P_{1}, \dots, P_{i-1}$.

Algorithm~\ref{alg:patience-sorting-mod1} below uses this characterization of piles.
The correctness follows directly from the characterization.
It uses a constant number of pointers of $O(\log n)$ bits
and a Boolean table of length $n$ for maintaining ``used'' and ``unused'' flags.
Thus it uses $n + O(\log n)$ bits working space in total.
The running time is $O(n^{2})$: each for-loop takes $O(n)$ time and the loop is repeated at most $n$ times.

\begin{algorithm}
  \caption{Computing $\lis(\pi)$ with $O(n)$ bits and in $O(n^{2})$ time}
  \label{alg:patience-sorting-mod1}
  \begin{algorithmic}[1]
    \State set $\ell := 0$ and mark all elements in $\pi$ as ``unused''
    \While{there is an ``unused'' element in $\pi$}
    \State increment $\ell$ and set $t := \infty$
    \For{$i=1$ \textbf{to} $n$} \Comment{\textcolor{gray}{this for-loop constructs the next pile implicitly}}
        \If{$\pi(i)$ is unused and $\pi(i) < t$}
            \State mark $\pi(i)$ as ``used'' and set $t := \pi(i)$ \Comment{\textcolor{gray}{$t$ is currently on top of $P_{\ell}$}}
	\EndIf
    \EndFor
    \EndWhile
    \State \Return $\ell$
  \end{algorithmic}
\end{algorithm}


\section{An algorithm for computing the length}
\label{sec:length}

In this section, we present our main algorithm
that computes $\lis(\sseq)$ with $O(s \log n)$ bits in $O(\frac{1}{s} \cdot n^{2} \log n)$ time for $\sqrt{n} \le s \le n$.
Note that the algorithm here outputs the length $\lis(\sseq)$ only.
The next section discusses efficient solutions to actually compute a longest sequence.

In the following, by $P_{i}$ for some $i$ we mean the $i$th pile obtained by (completely)
executing {\PS} unless otherwise stated.
(We sometimes refer to a pile at some specific point of the execution.)
Also, by $P_{i}(j)$ for $1 \le j \le |P_{i}|$ we denote the $j$th element added to $P_{i}$.
That is, $P_{i}(1)$ is the first element added to $P_{i}$ and $P_{i}(|P_{i}|)$ is the top element of $P_{i}$.

To avoid mixing up repeated elements,
we assume that each element $\sseq(j)$ of the piles is stored with its index $j$.
In the following, we mean by ``$\sseq(j)$ is in $P_{i}$'' that the $j$th element of $\sseq$ is pushed to $P_{i}$.
Also, by ``$\sseq(j)$ is $P_{i}(r)$'' we mean that the $j$th element of $\sseq$ is the $r$th element of $P_{i}$.

We start with an overview of our algorithm.
It scans over the input $O(n / s)$ times.
In each pass, it assumes that a pile $P_{i}$ with at most $s$ elements is given, which has been computed in the previous pass.
Using this pile $P_{i}$, it filters out the elements in the previous piles $P_{1}, \dots, P_{i-1}$.
It then basically simulates {\PS} but only in order to compute the next $2s$ piles.
As a result of the pass, it computes a new pile $P_{j}$ with at most $s$ elements such that $j \ge i + s$.

The following observation, that follows directly from the definition of {\PS} and Observation~\ref{obs:increasing-top}, 
will be useful for the purpose of filtering out elements in irrelevant piles.
\begin{observation}
\label{obs:filter}
Let $\sseq(y) \in P_{j}$ with $j \ne i$.
If $\sseq(x)$ was the top element of $P_{i}$ when $\sseq(y)$ was pushed to $P_{j}$,
then $j < i$ if $\sseq(y) < \sseq(x)$, and $j > i$ if $\sseq(y) > \sseq(x)$.
\end{observation}

Using Observation~\ref{obs:filter}, we can obtain the following algorithmic lemma
that plays an important role in the main algorithm.
\begin{lemma}
\label{lem:ignoring_piles}
Having stored $P_{i}$ explicitly in the additional memory and given an index $j > i$, 
the size $|P_{k}|$ for all $i+1 \le k \le \min\{j, \lis(\sseq)\}$ can be computed
in $O(n \log n)$ time with $O((|P_{i}| + j-i) \log n)$ bits.
If $\lis(\sseq) < j$, then we can compute $\lis(\sseq)$ in the same time and space complexity.
\end{lemma}
\begin{proof}
Recall that {\PS} scans the sequence $\sseq$ from left to right
and puts each element to the appropriate pile.
We process the input in the same way except that we filter out, and thereby ignore, the elements in
the piles $P_{h}$ for which $h <i$ or $h >j$.

To this end, we use the following two filters whose correctness follows from Observation~\ref{obs:filter}.

\emph{(Filtering~$P_h$ with~$h<i$.)} To filter out the elements that lie in $P_{h}$ for some $h <i$,
we maintain an index $r$ that points to the element of $P_{i}$ read most recently in the scan. Since $P_{i}$ is given explicitly to the algorithm, we can maintain such a pointer~$r$. 
 
When we read a new element $\sseq(x)$, we have three cases.
\begin{itemize}
  \item If $\sseq(x)$ is $P_{i}(r+1)$, then we increment the index $r$. 
  \item Else if $\sseq(x) < P_{i}(r)$, then $\sseq(x)$ is ignored since it is in $P_{h}$ for some $h <i$.
  \item Otherwise we have $\sseq(x) > P_{i}(r)$. In this case $\sseq(x)$ is in $P_{h}$ for some $h > i$.
\end{itemize}

\emph{(Filtering~$P_h$ with~$h>j$.)}  The elements in $P_{h}$ for $h > j$ can be filtered without maintaining additional information as follows.
Let again $\sseq(x)$ be the newly read element.
\begin{itemize}
  \item If no part of $P_{j}$ has been constructed yet, then $\sseq(x)$ is in $P_{h}$ for some $h \le j$.
  \item Otherwise, we compare $\sseq(x)$ and the element $\sseq(y)$ currently on the top of $P_{j}$.
  \begin{itemize}
    \item If $\sseq(x) > \sseq(y)$, then $\sseq(x)$ is in $P_{h}$ for some $h > j$, and thus ignored.
    \item Otherwise $\sseq(x)$ is in $P_{h}$ for some $h \le j$.
  \end{itemize}
\end{itemize}

We simulate {\PS} only for the elements that pass both filters above.
While doing so, we only maintain the top elements of the piles and additionally store the size of each pile. This requires at most~$O((j-i) \log n)$ space, as required by the statement of the lemma.
For details see Algorithm~\ref{alg:computing-size}.

The running time remains the same since we only need constant number of additional steps 
for each step in {\PS} to filter out irrelevant elements.
If $P_{j}$ is still empty after this process, we can conclude that
$\lis(\sseq)$ is the index of the newest pile constructed.
\end{proof}
\begin{algorithm}
  \caption{Computing $|P_{k}|$ for all~$k$ with $i+1 \le k \le \min\{j, \lis(\sseq)\}$ when $P_{i}$ is given}
  \label{alg:computing-size}
  \begin{algorithmic}[1]
    \State set $r := 0$ \Comment{\textcolor{gray}{$r$ points to the most recently read element in $P_{i}$}}
    \State set $\ell := i$ \Comment{\textcolor{gray}{the largest index of the piles constructed so far}}
    \State initialize $p_{i+1}, \dots, p_{j}$ to $\infty$ \Comment{\textcolor{gray}{$p_{k}$ is the element currently on top of $P_{k}$}}
    \State initialize $c_{i+1}, \dots, c_{j}$ to $0$ \Comment{\textcolor{gray}{$c_{k}$ is the current size of $P_{k}$}}
    \For{$x=1$ \textbf{to} $n$}
        \Statex $\triangleright$ \textcolor{gray}{filtering out irrelevant elements}
        \If{$\sseq(x)$ is $P_{i}(r+1)$}
	    \State increment $r$ and continue the for-loop
	\ElsIf{$\sseq(x) < P_{i}(r)$ or ($\ell \ge j$ and $\sseq(x) > p_{j}$)}
	    \State ignore the element and continue  the for-loop
	\EndIf
        \Statex $\triangleright$ \textcolor{gray}{push $\sseq(x)$ to the appropriate pile}
	\If{$\sseq(x) > p_{\ell}$}
            \State increment $\ell$ and set $h := \ell$
	\Else
	    \State set $h$ to be the smallest index with $\sseq(i) < p_{h}$
	\EndIf
	\State set $p_{h} := \sseq(x)$ and increment $c_{h}$
    \EndFor
  \end{algorithmic}
\end{algorithm}

The proof of Lemma~\ref{lem:ignoring_piles} can be easily adapted to also compute the pile~$P_j$ explicitly. 
For this, we simply additionally store all elements of $P_{j}$ as they are added to the pile.
\begin{lemma}
\label{lem:compute-a-small-pile}
Given $P_{i}$ and an index $j$ such that $i < j \le \lis(\sseq)$, we can compute $P_{j}$
in $O(n \log n)$ time with $O((|P_{i}| + |P_{j}| + j-i) \log n)$ bits.
\end{lemma}

Assembling the lemmas of this section, we now present our first main result.
The corresponding pseudocode of the algorithm 
can be found in Algorithm~\ref{alg:patience-sorting-2}.

\begin{theorem}
\label{thm:length}
There is an algorithm that, given an integer $s$ satisfying $\sqrt{n} \le s \le n$ and a sequence $\sseq$ of length $n$, computes $\lis(\sseq)$ 
in $O(\frac{1}{s} \cdot n^{2} \log n)$ time with $O(s \log n)$ bits of space.
\end{theorem}
\begin{proof}
To apply Lemmas~\ref{lem:ignoring_piles} and \ref{lem:compute-a-small-pile} at the beginning,
we start with a dummy pile $P_{0}$ with a single dummy entry $P_{0}(1) = -\infty$.
In the following, assume that for some $i \ge 0$ we computed the pile $P_{i}$ of size at most $s$ explicitly.
We repeat the following process until we find $\lis(\sseq)$.

In each iteration, we first compute the size $|P_{k}|$ for $i+1 \le k \le i+2s$.
During this process, we may find $\lis(\sseq) < i+2s$.
In such a case we output $\lis(\sseq)$ and terminate.
Otherwise, we find an index $j$ such that $i+s+1 \le j \le i+2s$ and $|P_{j}| \le n/s$.
Since $s \ge \sqrt{n}$, it holds that $|P_{j}| \le n/\sqrt{n} = \sqrt{n} \le s$.
We then compute $P_{j}$ itself to replace $i$ with $j$ and repeat.

By Lemmas~\ref{lem:ignoring_piles} and \ref{lem:compute-a-small-pile},
each pass can be executed in $O(n \log n)$ time with $O(s \log n)$ bits.
There are at most $\lis(\sseq) / s$ iterations, since in each iteration 
the index $i$ increases by at least $s$ or $\lis(\sseq)$ is determined.
Since $\lis(\sseq) \le n$, the total running time is $O(\frac{1}{s} \cdot n^{2} \log n)$.
\end{proof}
\begin{algorithm}
  \caption{Computing $\lis(\sseq)$ with $O(s \log n)$ bits in $O(\frac{1}{s} \cdot n^{2} \log n)$ time}
  \label{alg:patience-sorting-2}
  \begin{algorithmic}[1]
    \State set $i := 0$ and initialize the dummy pile $P_{0}$ with the single element $-\infty$
    \Loop
      \State compute the size of $P_{k}$ for all~$k$ with $i+1 \le k \le i+2s$
      \If{we find $\lis(\sseq) < i+2s$}
        \State \Return $\lis(\sseq)$
      \EndIf
      \State let $j$ be the largest index such that $|P_{j}| \le s$
      \Comment{\textcolor{gray}{$i + s + 1 \le j \le i + 2s$}}
      \State compute $P_{j}$ and set $i := j$
    \EndLoop
  \end{algorithmic}
\end{algorithm}

In the case of the smallest memory consumption we conclude the following corollary.

\begin{corollary}
Given a sequence $\sseq$ of length $n$,
$\lis(\sseq)$ can be computed in $O(n^{1.5} \log n)$ time with $O(\sqrt{n} \log n)$ bits of space.
\end{corollary}


\section{An algorithm for finding a longest increasing subsequence}
\label{sec:lis-sequence}

It is easy to modify the algorithm in the previous section
in such a way that it outputs an element of the final pile $P_{\lis(\sseq)}$,
which is the last element of a longest increasing subsequence by Proposition~\ref{prop:pile-characterization}.
Thus we can repeat the modified algorithm $n$ times
(considering only the elements smaller than and appearing before the last output)
and actually find a longest increasing subsequence.\footnote{%
This algorithm outputs a longest increasing subsequence in the reversed order.
One can access the input in the reversed order and find a longest \emph{decreasing} subsequence to avoid this issue.}
The running time of this na\"{\i}ve approach is $O(\frac{1}{s} \cdot n^{3} \log n)$.

As we claimed before, we can do much better.
In fact, we need only an additional multiplicative factor of $O(\log n)$ instead of $O(n)$ in the running time,
while keeping the space complexity as it is.
In the rest of this section, we prove the following theorem.

\begin{theorem}
\label{thm:subsequence}
There is an algorithm that, given an integer $s$ satisfying $\sqrt{n} \le s \le n$ and a sequence $\sseq$ of length $n$,  computes
a longest increasing subsequence of $\sseq$
 in $O(\frac{1}{s} \cdot n^{2} \log^{2} n)$ time using $O(s \log n)$ bits of space.
\end{theorem}
\begin{corollary}
Given a sequence $\sseq$ of length $n$,
a longest increasing subsequence of $\sseq$
can be found in $O(n^{1.5} \log^{2} n)$ time with $O(\sqrt{n} \log n)$ bits of space.
\end{corollary}

We should point out that the algorithm in this section is \emph{not} a multi-pass algorithm.
However, we can easily transform it without any increase in the time and space complexity
so that it works as a sequential access algorithm.

\subsection{High-level idea}
We first find an element that is in a longest increasing subsequence roughly in the middle.
As we will argue, this can be done in $O(\frac{1}{s} \cdot n^{2} \log n)$ time with $O(s \log n)$ bits
by running the algorithm from the previous section twice, once in the ordinary then once in the reversed way. 
We then divide the input into the left and right parts at a near-mid element and recurse.

The space complexity remains the same and 
the time complexity increases only by an $O(\log n)$ multiplicative factor.
The depth of recursion is $O(\log n)$ and at each level of recursion the total running time is
$O(\frac{1}{s} \cdot n^{2} \log n)$.
To remember the path to the current recursion, 
we need some additional space, but it is bounded by $O(\log^{2} n)$ bits.

\subsection{A subroutine for short longest increasing sequences}
We first solve the base case in which $\lis(\sseq) \in O(n/s)$.
In this case, we use the original {\PS} and repeat it $O(n/s)$ times.
We present the following general form first.
\begin{lemma}\label{lem:base:case}
Let $\sseq$ be a sequences of length $n$ and $\lis(\sseq) = k$.
Then a longest increasing subsequence of $\sseq$ 
can be found in $O(k \cdot n \log k)$ time
with $O(k \log n)$ bits.
\end{lemma}
\begin{proof}
Without changing the time and space complexity,
we can modify the original {\PS} so that
\begin{itemize}
  \item it maintains only the top elements of the piles;
  \item it ignores the elements larger than or equal to a given upper bound; and
  \item it outputs an element in the final pile.
\end{itemize}

We run the modified algorithm $\lis(\sseq)$ times.
In the first run, we have no upper bound.
In the succeeding runs, we set the upper bound to be the output of the previous run.
In each run the input to the algorithm is the initial part of the sequence that ends right before the last output.
The entire output forms a longest increasing sequence of $\sseq$.\footnote{%
Again this output is reversed. 
We can also compute the output in nonreversed order as discussed before.}

Since $\lis(\sseq) = k$, modified {\PS} maintains only $k$ piles.
Thus each run takes $O(n \log k)$ time and uses $O(k \log n)$ bits.
The lemma follows since this is repeated $k$ times and
 each round only stores $O(\log n)$ bits of information from the previous round.
\end{proof}

The following special form of the lemma above holds
since $n /s \le s$ when $s \ge \sqrt{n}$.
\begin{corollary}
Let $\sseq$ be a sequence of length $n$ and $\lis(\sseq) \in O(n/s)$ for some $s$ with $\sqrt{n} \le s \le n$.
A longest increasing subsequence of $\sseq$ 
can be found in $O(\frac{1}{s} \cdot n^{2} \log n)$ time with $O(s \log n)$ bits.
\end{corollary}

\subsection{A key lemma}

As mentioned above, we use a reversed version of our algorithm.
{\RPS} is the reversed version of {\PS}:
it reads the input from right to left and uses the reversed inequalities.
 (See Algorithm~\ref{alg:reversed-patience-sorting}.)
{\RPS} computes the length of a longest decreasing subsequence in the reversed sequence,
which is a longest increasing subsequence in the original sequence.
Since the difference between the two algorithms is small,
we can easily modify our algorithm in Section~\ref{sec:length}
for the length so that it simulates {\RPS} instead of {\PS}.
\begin{algorithm}
  \caption{{\RPS}}
  \label{alg:reversed-patience-sorting}
  \begin{algorithmic}[1]
    \State set $\ell := 0$ and initialize the dummy pile $Q_{0}$ with the single element $+\infty$
    \For{$i=n$ \textbf{to} $1$} 
        \If{$\sseq(i) < \mathtt{top}(Q_{\ell})$}
            \State increment $\ell$, let $Q_{\ell}$ to be a new empty pile, and set $j := \ell$ 
	\Else
	    \State set $j$ to be the smallest index with $\sseq(i) > \mathtt{top}(Q_{j})$
	\EndIf
	\State push $\sseq(i)$ to $Q_{j}$
    \EndFor
    \State \Return $\ell$
  \end{algorithmic}
\end{algorithm}

Let $Q_{i}$ be the $i$th pile constructed by {\RPS} as in Algorithm~\ref{alg:reversed-patience-sorting}.
Using Proposition~\ref{prop:pile-characterization},
we can show that for each $\sseq(i)$ in $Q_{j}$,
the longest decreasing subsequence of the reversal of $\sseq$ ending at $\sseq(i)$ has length $j$.
This is equivalent to the following observation.
\begin{observation}
  $\sseq(i) \in Q_{j}$ if and only if
  a longest increasing subsequence of $\sseq$ starting at $\sseq(i)$ has length $j$.
\end{observation}
This observation immediately gives the key lemma below.
\begin{lemma}
$P_{k} \cap Q_{\lis(\sseq) - k + 1} \ne \emptyset$ for all~$k$ with $1 \le k \le \lis(\sseq)$.
\end{lemma}
\begin{proof}
Let $\seq{\sseq(i_{1}), \dots, \sseq(i_{\ell})}$ be a longest increasing subsequence of $\sseq$.
Proposition~\ref{prop:pile-characterization} implies that $\sseq(i_{k}) \in P_{k}$.
The subsequence $\seq{\sseq(i_{k}), \dots, \sseq(i_{\ell})}$
 is a longest increasing subsequence of $\sseq$ starting at $\sseq(i_{k})$
since otherwise $\seq{\sseq(i_{1}), \dots, \sseq(i_{\ell})}$ is not longest.
Since the length of $\seq{\sseq(i_{k}), \dots, \sseq(i_{\ell})}$ is $i_{\ell} - k + 1 = \lis(\sseq) - k + 1$,
we have $\sseq(k) \in Q_{\lis(\sseq) - k + 1}$.
\end{proof}

Note that the elements of $P_{k}$ and $Q_{\lis(\sseq) - k + 1}$ are not the same in general.
For example, by applying $\RPS$ to $\sseq_{1} = \seq{2, 8, 4, 9, 5, 1, 7, 6, 3}$,
we get $Q_{1} = \seq{3,6,7,9}$, $Q_{2} = \seq{1,5,8}$, $Q_{3} = \seq{4}$, and $Q_{4} = \seq{2}$ as below. 
(Recall that $P_{1} = \seq{2,1}$, $P_{2} = \seq{8,4,3}$, $P_{3} = \seq{9,5}$, and $P_{4} = \seq{7,6}$.)
The following diagram depicts the situation. The elements shared by $P_{k}$ and $Q_{\lis(\sseq) - k + 1}$ are colored and underlined.

\smallskip

{
  \small
  \tabcolsep=0.5mm
  \noindent
  \begin{tabular}{c}
    \\
    \\
    \\
    \textbf{3} \\
    \hline
    $Q_{1}$
  \end{tabular}
  \hfill
  \begin{tabular}{c}
    \\
    \\
    \textbf{6} \\
    3 \\
    \hline
    $Q_{1}$
  \end{tabular}
  \hfill
  \begin{tabular}{c}
    \\
    \textbf{7} \\
    6 \\
    3 \\
    \hline
    $Q_{1}$
  \end{tabular}
  \hfill
  \begin{tabular}{cc}
    \\
    7 \\
    6 \\
    3 & \textbf{1} \\
    \hline
    $Q_{1}$ & $Q_{2}$
  \end{tabular}
  \hfill
  \begin{tabular}{cc}
    \\
    7 \\
    6 & \textbf{5} \\
    3 & 1 \\
    \hline
    $Q_{1}$ & $Q_{2}$
  \end{tabular}
  \hfill
  \begin{tabular}{cc}
    \textbf{9} \\
    7 \\
    6 & 5 \\
    3 & 1 \\
    \hline
    $Q_{1}$ & $Q_{2}$
  \end{tabular}
  \hfill
  \begin{tabular}{ccc}
    9 \\
    7 \\
    6 & 5 \\
    3 & 1 & \textbf{4} \\
    \hline
    $Q_{1}$ & $Q_{2}$ & $Q_{3}$
  \end{tabular}
  \hfill
  \begin{tabular}{ccc}
    9 \\
    7 & \textbf{8} \\
    6 & 5 \\
    3 & 1 & 4 \\
    \hline
    $Q_{1}$ & $Q_{2}$ & $Q_{3}$
  \end{tabular}
  \hfill
  \begin{tabular}{cccc}
    9 \\
    \textcolor{vr}{\underline{7}} & 8 \\
    \textcolor{vr}{\underline{6}} & \textcolor{vr}{\underline{5}} \\
    3 & 1 & \textcolor{vr}{\underline{4}} & \textcolor{vr}{\textbf{\underline{2}}} \\
    \hline
    $Q_{1}$ & $Q_{2}$ & $Q_{3}$ & $Q_{4}$
  \end{tabular}
  \hfill\hfill\hfill\hfill
  \begin{tabular}{cccc}
    \\
      & 3 &   &   \\
    1 & \textcolor{vr}{\underline{4}} & \textcolor{vr}{\underline{5}} & \textcolor{vr}{\underline{6}} \\
    \textcolor{vr}{\underline{2}} & 8 & 9 & \textcolor{vr}{\underline{7}} \\
    \hline
    $P_{1}$ & $P_{2}$ & $P_{3}$ & $P_{4}$
  \end{tabular}
}

\subsection{The algorithm}

We first explain the subroutine for finding a near-mid element in a longest increasing subsequence.
\begin{lemma}
Let $s$ be an integer satisfying $\sqrt{n} \le s \le n$.
Given a sequence $\sseq$ of length $n$,
the $k$th element of a longest increasing subsequence of $\sseq$
for some $k$ with $\lis(\sseq)/2 \le k < \lis(\sseq)/2 + n/s$
can be found in $O(\frac{1}{s} \cdot n^{2} \log n)$ time using $O(s \log n)$ bits of space.
\end{lemma}
\begin{proof}
We slightly modify Algorithm~\ref{alg:patience-sorting-2} 
so that it finds an index $k$ and outputs $P_{k}$ such that $|P_{k}| \le s$ and 
$\lis(\sseq)/2 \le k \le \lis(\sseq)/2 + n/s$.
Such a $k$ exists since the average of $|P_{i}|$ for 
$\lis(\sseq)/2 \le i < \lis(\sseq)/2 + n/s$ is at most $s$.
The time and space complexity of this phase are as required by the lemma.

We now find an element in $P_{k} \cap Q_{\lis(\sseq) - k + 1}$.
Since the size $|Q_{\lis(\sseq) - k + 1}|$ is not bounded by~$O(s)$ in general,
we cannot store $Q_{\lis(\sseq) - k + 1}$ itself.
Instead use the reversed version of the algorithm in Section~\ref{sec:length} to enumerate it.
Each time we find an element in $Q_{\lis(\sseq) - k + 1}$,
 we check whether it is included in $P_{k}$.
This can be done with no loss in the running time since $P_{k}$ is sorted and
the elements of $Q_{\lis(\sseq) - k + 1}$ arrive in increasing order.
\end{proof}

The next technical but easy lemma allows us to split the input into two parts at an element of a longest increasing subsequence
and to solve the smaller parts independently. 
\begin{lemma}
\label{lem:recurse}
Let $\sseq(j)$ be the $k$th element of a longest increasing subsequence of a sequence~$\sseq$.
Let $\sseq_{L}$ be the subsequence of $\sseq\subseq{1,\dots,j-1}$
formed by the elements smaller than $\sseq(j)$.
Similarly let $\sseq_{R}$ be the subsequence of $\sseq\subseq{j+1,\dots,|\sseq|}$
formed by the elements larger than $\sseq(j)$.
Then, a longest increasing subsequence of $\sseq$ can be obtained by concatenating
a longest increasing subsequence of $\sseq_{L}$, $\sseq(j)$, and
a longest increasing subsequence of $\sseq_{R}$, in this order.
\end{lemma}
\begin{proof}
Observe that the concatenated sequence is an increasing subsequence of $\sseq$.
Thus it suffices to show that $\lis(\sseq_{L}) + \lis(\sseq_{R}) + 1 \ge \lis(\sseq)$.
Let $\sseq\subseq{i_{1}, \dots, i_{\lis(\sseq)}}$ be 
a longest increasing subsequence of $\sseq$ such that $i_{k} = j$.
From the definition, $\sseq\subseq{i_{1}, \dots, i_{k-1}}$ is a subsequence of $\sseq_{L}$,
and $\sseq\subseq{i_{k+1}, \dots, i_{\lis(\sseq)}}$ is a subsequence of $\sseq_{R}$.
Hence $\lis(\sseq_{L}) \ge k-1$ and $\lis(\sseq_{R}) \ge \lis(\sseq)-k$,
and thus $\lis(\sseq_{L}) + \lis(\sseq_{R}) + 1 \ge \lis(\sseq)$.
\end{proof}
As Lemma~\ref{lem:recurse} suggests, after finding a near-mid element $\sseq(k)$, 
we recurse into $\sseq_{L}$ and $\sseq_{R}$.
If the input $\sseq'$ to a recursive call has small $\lis(\sseq')$, 
we directly compute a longest increasing subsequence.
See Algorithm~\ref{alg:actual-lis} for details of the whole algorithm.
Correctness follows from Lemma~\ref{lem:recurse}
and correctness of the subroutines.
\begin{algorithm}
  \caption{Recursively finding a longest increasing subsequence of $\rho$}
  \label{alg:actual-lis}
  \begin{algorithmic}[1]
    \State \Call{RecursiveLIS}{$\rho$, $-\infty$, $+\infty$}
    \Procedure{RecursiveLIS}{$\sseq$, $\mathsf{lb}$, $\mathsf{ub}$}
    \State $\sseq' := $ the subsequence of $\sseq$ formed by the elements $\sseq(i)$ such that $\mathsf{lb} < \sseq(i) < \mathsf{ub}$ \label{alg:actual-lis:sigmaprime}
    \Statex \Comment{\textcolor{gray}{$\sseq'$ is not explicitly computed but provided by ignoring the irrelevant elements}}
    
    \State compute $\lis(\sseq')$
    
    \If{$\lis(\sseq') \le 3 |\sseq'| / s$}
        \State \textbf{output} a longest increasing subsequence of $\sseq'$\Comment{\textcolor{gray}{Lemma~\ref{lem:base:case}}}
    \Else
        \State find the $k$th element $\sseq'(j)$ of a longest increasing subsequence of $\sseq'$ 
	\label{alg:actual-lis:k}
	\Statex \hspace{2cm} for some $k$ with $\lis(\sseq')/2 \le k < \lis(\sseq')/2 + |\sseq'|/s$
        \State \Call{RecursiveLIS}{$\sseq'\subseq{1,\dots,j-1}$, $\mathsf{lb}$, $\mathsf{\sseq'(j)}$}
        \State \textbf{output} $\sseq'(j)$
        \State \Call{RecursiveLIS}{$\sseq'\subseq{j+1,\dots,|\sseq'|}$, $\sseq'(j)$, $\mathsf{ub}$}
    \EndIf
    \EndProcedure
  \end{algorithmic}
\end{algorithm}

\subsection{Time and space complexity}

In Theorem~\ref{thm:subsequence}, the claimed running time is $O(\frac{1}{s} \cdot n^{2} \log^{2} n)$.
To prove this, we first show that the depth of the recursion is $O(\log n)$.
We then show that the total running time in each recursion level is $O(\frac{1}{s} \cdot n^{2} \log n)$.
The claimed running time is guaranteed by these bounds.

\begin{lemma}
\label{lem:recursion-depth}
Given a sequence $\sseq$, the depth of the recursions invoked
by \textsc{RecursiveLIS} of Algorithm~\ref{alg:actual-lis} is at most $\log_{6/5} \lis(\sseq')$,
where $\sseq'$ is the subsequence of $\sseq$ computed in Line~\ref{alg:actual-lis:sigmaprime}.
\end{lemma}
\begin{proof}
We proceed by induction on $\lis(\sseq')$.
If $\lis(\sseq') \le 3 |\sseq'|/s$, then no recursive call occurs,
and hence the lemma holds.
In the following, we assume that $\lis(\sseq') = \ell > 3 |\sseq'|/s$
and that the statement of the lemma is true for any sequence $\sseq''$ with $\lis(\sseq'') < \ell$.

Since $\ell > 3 |\sseq'|/s$, we recurse into two branches on subsequences of $\sseq'$.
From the definition of $k$ in Line~\ref{alg:actual-lis:k} of Algorithm~\ref{alg:actual-lis},
the length of a longest increasing subsequence is less than $\ell/2 + |\sseq'|/s$ in each branch.
Since $\ell/2 + |\sseq'|/s < \ell/2 + \ell/3 = 5\ell/6$,
each branch invokes recursions of depth at most $\log_{6/5} (5\ell/6) = \log_{6/5} \ell - 1$.
Therefore the maximum depth of the recursions invoked by their parent is at most $\log_{6/5} \ell$.
\end{proof}

\begin{lemma}
Given a sequence $\sseq$ of length $n$, 
the total running time at each depth of recursion excluding further recursive calls 
in Algorithm~\ref{alg:actual-lis} takes $O(\frac{1}{s} n^{2} \log n)$ time.
\end{lemma}
\begin{proof}
In one recursion level, we have many calls of \textsc{RecursiveLIS}
on pairwise non-overlapping subsequences of $\sseq$.
For each subsequence $\sseq'$, 
the algorithm spends time $O(\frac{1}{s} |\sseq'|^{2} \log |\sseq'|)$.
Thus the total running time at a depth is 
$O(\sum_{\sseq'} \frac{1}{s} |\sseq'|^{2} \log |\sseq'|)$,
which is $O(\frac{1}{s} n^{2} \log n)$ since $\sum_{\sseq'} |\sseq'|^{2} \le |\sseq|^{2} = n^{2}$.
\end{proof}

Finally we consider the space complexity of Algorithm~~\ref{alg:actual-lis}.
\begin{lemma}
Algorithm~\ref{alg:actual-lis} uses $O(s \log n)$ bits of working space on sequences of length~$n$.
\end{lemma}
\begin{proof}
We have already shown that each subroutine uses $O(s \log n)$ bits.
Moreover, this space of working memory can be discarded before another subroutine call occurs.
Only a constant number of $O(\log n)$-bit words are passed to the new subroutine call.
We additionally need to remember the stack trace of the recursion.
The size of this additional information is bounded by $O(\log^{2} n)$ bits
since each recursive call is specified by a constant number of $O(\log n)$-bit words
and the depth of recursion is $O(\log n)$ by Lemma~\ref{lem:recursion-depth}.
Since $\log^{2} n \in O(s \log n)$ for $s \ge \sqrt{n}$, the lemma holds.
\end{proof}


\section{Lower bound for algorithms with sequential access}
\label{sec:lb}

An algorithm is a \emph{sequential access} algorithm
if it can access elements in the input array only sequentially. In our situation this means that for a given sequence, accessing the~$i$th element of the sequence directly after having accessed the~$j$th element of the sequence costs time at least linear in~$|i-j|$. As opposed to the RAM, any Turing machine in which the input is given on single read-only tape has this property. 
Note that any lower bound for sequential access algorithms in an asymptotic form is applicable to multi-pass algorithms as well
since every multi-pass algorithm can be simulated by a sequential access algorithm with the same asymptotic behavior.
Although some of our algorithms are not multi-pass algorithms,
it is straightforward to transform them to sequential access algorithms with the same time and space complexity.

To show a lower bound on the running time of sequential access algorithms with limited working space,
we need the concept of communication complexity (see \cite{KushilevitzN1997commucomp} for more details).
Let $f$ be a function.
Given $\alpha \in \mathcal{A}$ to the first player Alice and $\beta \in \mathcal{B}$ to the second player Bob,
the players want to compute $f(\alpha, \beta)$ together by sending bits to each other (possibly multiple times).
The communication complexity of $f$ is the maximum number of bits transmitted between Alice and Bob
over all inputs by the best protocol for $f$.

Consider the following variant of the LIS problem: Alice gets the first half of a permutation
$\pi$ of $\{1,\dots,2n\}$ and Bob gets the second half. They compute $\lis(\pi)$ together.
It is known that this problem has high communication 
complexity~\cite{Liben-NowellVZ2006lis,GopalanJKK2007estimating,SunW2007communication}.
\begin{proposition}
[\cite{GopalanJKK2007estimating,SunW2007communication}]
\label{prop:comlb'}
Let $\pi$ be a permutation of $\{1,\dots,2n\}$.
Given the first half of $\pi$ to Alice and the second half to Bob,
they need $\Omega(n)$ bits of communication to compute $\lis(\pi)$ in the worst case
(even with 2-sided error randomization).
\end{proposition}

Now we present our lower bound.
Note that the lower bound even holds for the special case where input is restricted to permutations.
\begin{theorem}
\label{thm:lowerbound}
Given a permutation $\pi$ of $\{1,\dots,4n\}$,
any sequential access (possibly randomized) algorithm computing $\lis(\pi)$ using $b$ bits
takes $\Omega(n^{2}/b)$ time.
\end{theorem}
\begin{proof}
Given an arbitrary~$n>1$, let~$\pi'$ be a permutation of~$\{1,\dots,2n \}$.
We construct a permutation~$\pi$ of~$\{1,\dots,4n\}$ as follows. 
Let~$\pi'_1 = \seq{\pi(1),\dots,\pi(n)}$ be the first half of~$\pi'$,
define~$\pi'_2 = \seq{4n,4n-1,\dots, 2n+2}$ and
let~$\pi'_3 =\seq{\pi(n+1),\pi(n+2),\ldots,\pi(2n)}$ be the second half of~$\pi'$.
Then we define~$\pi$ to be the concatenation of~$\pi'_1$,~$\pi'_2$,~$\pi'_3$ and
the one element sequence~$\pi'_4 =\seq{2n+1}$, in that order.

It is not difficult to see that~$\pi$ is a permutation and that~$\lis(\pi) = \lis(\pi')+1$.
To see the latter, observe that the concatenation of $\pi'_{2}$ and $\pi'_{4}$ is a decreasing subsequence of $\pi$.
Hence any increasing subsequence of $\pi$ can contain at most one element not in $\pi'$.
On the other hand, any increasing subsequence of $\pi'$ of length $\ell$
can be extended with the element $2n+1$ of $\pi'_{4}$ to an increasing subsequence of $\pi$ of length $\ell + 1$.

We say a sequential access algorithm traverses the middle if it accesses a position in~$\pi'_1$ and then accesses a position in~$\pi'_3$ or vice versa with possibly accessing elements in~$\pi'_2$ but only such elements in meantime. Since each traversal of the middle takes $\Omega(n)$ time, it suffices to show that the number of traversals of the middle is $\Omega(n/b)$.

Suppose we are given a sequential access algorithm $M$ that computes~$\lis(\pi)$ with $t$ traversals of the middle. 
Using $M$, we construct a two-player communication protocol for computing $\lis(\pi')$ with at most $t b$ bits of communication. (A similar technique is described for streaming algorithms in \cite{SunW2007communication}.)

Recall that the first player Alice gets the first half $\pi'_{1}$ of $\pi'$
and the second player Bob gets the second half $\pi'_{3}$ of $\pi'$.
They compute $\lis(\pi')$ together as follows.
\begin{itemize}
  \item Before starting computation, Alice computes $\pi_{A}$ by concatenating $\pi'_{1}$ and $\pi'_{2}$ in that order,
  and Bob computes $\pi_{B}$ by concatenating $\pi'_{2}$, $\pi'_{3}$, and $\pi'_{4}$ in that order.
  
  \item They first compute $\lis(\pi)$ using $M$ by repeating the following phases:
  \begin{itemize}
    \item Alice starts the computation by $M$ and continues while $M$ stays in $\pi\subseq{1,\dots,3n-1} = \pi_{A}$.
    When $M$ tries to access $\pi\subseq{3n,\dots,4n}$, and thus a traversal of the middle occurs,
    Alice stops and sends all $b$ bits stored by $M$ to Bob.

    \item Bob restores the $b$ bits received from Alice to the working memory of $M$
    and continues computation while $M$ stays in $\pi\subseq{n+1,\dots,4n} = \pi_{B}$.
    A traversal of the middle is occurred when $M$ tries to access $\pi\subseq{1,\dots,n}$.
    Bob then stops and sends the $b$ bits currently stored by $M$ back to Alice.
  \end{itemize}
  
  \item When $M$ outputs $\lis(\pi)$ and terminates, the currently active player outputs $\lis(\pi)-1$ as $\lis(\pi')$
  and terminates the protocol.
\end{itemize}
The two players correctly simulate $M$ and, as a result, compute $\lis(\pi')$ together.
Since the algorithm $M$ invokes $t$ traversals, the total number of bits sent is at most $t b$.
Since $tb \in \Omega(n)$ holds by Proposition~\ref{prop:comlb'}, we have $t \in \Omega(n/b)$ as required.
\end{proof}

Recall that our algorithms for the LIS problem
use $O(s \log n)$ bits and runs in $O(\frac{1}{s} n^{2} \log n)$ time for computing the length
and in $O(\frac{1}{s} n^{2} \log^{2} n)$ time for finding a subsequence, where $\sqrt{n} \le s \le n$.
By Theorem~\ref{thm:lowerbound}, their time complexity is
optimal for algorithms with sequential access up to polylogarithmic factors of $\log^{2} n$ and $\log^{3} n$, 
respectively.


\section{Concluding remarks}
Our result raises the following question: ``Do $o(\sqrt{n})$-space polynomial-time algorithms for LIS exist?''
An unconditional `no' answer would be surprising as
it implies $\mathrm{SC} \ne \mathrm{P} \cap \mathrm{PolyL}$,
where $\mathrm{SC}$ (Steve's Class) is the class of problems that can be solved by an algorithm that simultaneously
runs in polynomial-time and polylogarithmic-space~\cite{Cook1979deterministic,Nisan1992RLvsSC}.
A possibly easier question asks for the existence of a log-space algorithm.
For this question, one might be able to give some evidence for a `no' answer by showing
NL-hardness of (a decision version of) LIS.\@

We would like to mention some known results that have a mysterious coincidence in space complexity with our results.
For $(1+\epsilon)$-approximation of $\lis(\pi)$ by one-pass streaming algorithms, it is known that
$O(\sqrt{n / \epsilon} \cdot \log n)$ bits are sufficient~\cite{GopalanJKK2007estimating} 
and $\Omega(\sqrt{n/ \epsilon})$ bits are necessary~\cite{ErgunJ2015monotonicity,GalG2010lower}.
We were not able to find any connection here
and do not claim anything concrete about this coincidence.

To make the presentation simple, we used $n$ to bound $\lis(\sseq)$
in the time complexity analyses of the algorithms.
If we carefully analyze the complexity in terms of $\lis(\sseq)$ instead of $n$ when possible,
we can obtain the following \emph{output-sensitive} bounds.
\begin{theorem}
Let $s$ be an integer satisfying $\sqrt{n} \le s \le n$,
and let $\sseq$ be a sequence of length $n$ with $\lis(\sseq) = k$.
Using $O(s \log n)$ bits of space,
$\lis(\sseq)$ can be computed in $O(\frac{1}{s} \cdot k n \log k)$ time 
and 
a longest increasing subsequence of $\sseq$
can be found in $O(\frac{1}{s} \cdot k n \log^{2} k)$ time.
\end{theorem}


\bibliographystyle{plainurl}
\bibliography{lis}

\end{document}